\documentclass[journal]{IEEEtran}

\ifCLASSINFOpdf
\else
\fi

\usepackage{url}
\usepackage{amsmath,amssymb}
\usepackage{float}
\usepackage{graphicx,subfigure}
\usepackage{color}
\usepackage[colorlinks=true,linkcolor=black]{hyperref}
\usepackage{breakurl}
\usepackage{cite}
\usepackage[ruled,linesnumbered]{algorithm2e}
\usepackage{ amssymb }
\usepackage{mathrsfs }
\usepackage{enumitem} 
\usepackage{romannum}
\usepackage{hyperref}
\usepackage{amsthm}
\hypersetup{colorlinks,linkcolor={black},citecolor={black},urlcolor={red}}  

\newcommand{\Revised}{\textcolor{black}}
\newcommand{\ReviseTase}{\textcolor{black}}

\newtheorem{theorem}{Theorem}
\newtheorem{lemma}{Lemma}
\newtheorem{remark}{Remark}
\newtheorem{definition}{Definition}

\newtheorem{corollary}{Corollary}
\newtheorem{problem}{Problem}

\begin{document}

\title{\LARGE\bf Strategies to Inject Spoofed Measurement Data\\ to Mislead Kalman Filter}%
\author{Zhongshun~Zhang$^{1}$, 
        Lifeng~Zhou$^{2}$,
        and~Pratap~Tokekar$^{1}$
\thanks{$^{1}$Zhongshun Zhang and Pratap Tokekar are with the Department of Computer Science, University of Maryland, College Park, USA {\tt\small \{zszhang, tokekar\}}@umd.edu}%
\thanks{$^{2}$Lifeng Zhou is  with the Department of Electrical \& Computer Engineering, Virginia Tech, USA. {\tt\small lfzhou}@vt.edu}%
\thanks{This material is based upon work supported by the National Science Foundation under Grant \#1566247.}
\thanks{A preliminary version of this paper, without the evaluation with $\chi^2$ detector in Section~\ref{sec:detector}, was presented at ACC 2018~\cite{zhang2018strategies}.}}

\maketitle
\thispagestyle{empty}
\pagestyle{empty}

\begin{abstract}

We study the problem of designing false measurement data that is injected to corrupt and mislead the output of a Kalman filter. Unlike existing works that focus on detection and filtering algorithms for the observer, we study the problem from the attacker's point-of-view. In our model, the attacker can corrupt the measurements by injecting additive spoofing signals. The attacker seeks to create a separation between the estimate of the Kalman filter with and without spoofed signals. We present a number of results on how to inject spoofing signals while minimizing the magnitude of the injected signals. The resulting strategies are evaluated through simulations along with theoretical proofs. We also evaluate the spoofing strategy in the presence of a $\chi^2$ spoof detector. Building on our main result, we present a strategy that is proven to successfully mislead a Kalman filter while ensuring it is not detected.
\end{abstract}

\ReviseTase{
\emph{Note to Practitioners}---This paper is motivated by the need for understanding the limits of how much an adversary can corrupt and mislead the output of a state estimator. We study the problem from the adversary's point-of-view. We show how an adversary can mislead the output of a Kalman filter by injecting a minimal amount of energy without being detected by a residual-based attack/failure detector. The proposed approach can be applied as an attack strategy to steer the output of a Kalman filter by a desired amount. This is applicable in a variety of scenarios such as target tracking and robot localization. We point out some limitations of the residual-based attack/failure detectors. We expect that this work will lead to a better and more robust state estimator designs in the future.}

\section{Introduction}
As autonomous systems proliferate, there are growing concerns about their security and safety~\cite{parkinson2017cyber,thing2016autonomous}. Of particular concern is their vulnerability to signal spoofing attacks~\cite{tippenhauer2011requirements}. As a result, many researchers are designing algorithms that enable an \emph{observer} to detect and mitigate signal spoofing attacks (e.g.,~\cite{al2015design,chen2007detecting,gil2017guaranteeing,zhang2017functional,fan2017synchrophasor}). We study the problem from the opposite (i.e., the attacker's) point-of-view. Our goal is to characterize the capabilities of the attacker that is generating the spoofing signals while assuming that the observer is using a Kalman filter for state estimation. 

The problem of generating spoofing attacks has been studied specifically for GPS signals. Tippenhauer et al.~\cite{tippenhauer2011requirements} describe the requirements as well as present a methodology for generating spoofed GPS signals. Larcom and Liu~\cite{larcom2013modeling} presented a taxonomy of GPS spoofing attacks.

The typical approach to mitigate sensor spoofing attacks is by designing robust state estimators~\cite{bezzo2014attack}. Fawzi et al.~\cite{fawzi2011secure} presented the design of a state estimator for a linear dynamical system when some of the sensor measurements are corrupted by an adversarial attacker. We focus on the scenario where the observer uses a Kalman Filter (KF) for estimating the state using measurements that are corrupted by additive spoofing signals by the attackers. We study the problem of generating spoofing signals of minimum energy that can achieve any desired separation between the KF estimate with spoofing and without spoofing. We show that for many practical cases, the spoofing signals can be generated using linear programming in polynomial time. 

\ReviseTase{There has been recent work on designing spoofing attacks for specific systems such as a Linear Quadratic Gaussian controller~\cite{mo2010control}, GPS data~\cite{su2016stealthy}, wireless sensor networks~\cite{mo2010false}, and electric power grids~\cite{liu2011false}.}
In~\cite{liu2011false}, the authors present false data injection attacks, against state estimation in electric
power grids. This paper shows that an attacker can exploit the configuration of a power system to launch such attacks to successfully introduce arbitrary errors into certain state variables while bypassing existing techniques for bad measurement detection.

Another work by Su et al.~\cite{su2016stealthy} is closely related to ours. The authors show how to spoof the GPS signal without triggering a detector that uses the residual in the Kalman filter. They present a 1-step (greedy) online spoofing strategy that solves a linear relaxation of a Quadratically Constrained Quadratic Program (QCQP) at each timestep. We present a strategy that plans for $T$ future timesteps, instead of just the next timestep, while minimizing the spoofing signal energy. Furthermore, we characterize the scenarios under which our strategy finds the optimal solution in polynomial time.

The work that is most closely related to ours is by Mo et al.~\cite{mo2010control,mo2010false}. Their goal is to design false measurement data to mislead a system with Kalman filter~\cite{mo2010false} or an LQG control system~\cite{mo2010control}. Both, our work and the aforementioned work, assume that the system is linear with Gaussian noise and that a discrete Kalman filter is used to estimate the state. \Revised{Mo et al.~\cite{mo2010control} define $(\epsilon,\alpha)$--attacks (see Definitions 2 and 3 in the original paper). Based on their definition, an attack sequence is successful if: (1) the difference in the estimated state of the system under attack and the true state is greater than a given value; and (2) the probability of alarm for the $\chi^2$ failure detector is always smaller than a given threshold. 
They proved that a linear control system is perfectly attackable if and only if its transition matrix has an unstable eigenvalue  and the corresponding eigenvector satisfies additional conditions (c.f. Theorem 2 in \cite{mo2010control}). These conditions may be too strict. 
We relax the requirements of an attackable system with the goal of being applicable to more classes of systems. Specifically, we remove the second condition in the $(\epsilon,\alpha)$--attacks and instead we only consider minimizing the total injected signal by $\sum^{T}_{t=1} \gamma_t \cdot \lVert \epsilon_t\lVert_p^p$. Nevertheless, we also show the conditions under which an attack is successful against an $\chi^2$ detector.} 

\Revised{Bai et al. presented a different notion of a successful attack in two relevant papers~\cite{bai2017kalman,bai2017data}. They define a successful attack as $\epsilon$--stealthy.\footnote{This $\epsilon$ is not related to the $\epsilon$ used by Mo et al.~\cite{mo2010false,mo2010control}.} The goal is to maximize the mean squared error between the attacker's estimated state  and the true state subject to $\lim_{t \rightarrow \infty} \frac{1}{t} D(\tilde{r}_1^t || {r}_1^t) < \epsilon$. Here, $D(\tilde{r}_1^t || {r}_1^t)$ is the Kullback-Leibler divergence (KLD) between the innovation without attack, $r_1^t$, and with attack, $\tilde{r}_1^t$. Their notion of a successful attack only applies when $t \rightarrow \infty $. Thus, their attack strategy can only be applied when a Kalman filter runs for a long time. Instead, we focus on finite, possibly small, number of time steps and do not require $t \rightarrow \infty$. Furthermore, our notion of a successful attack differs from theirs and does not focus on a specific type of detector.}

Various failure detectors have been proposed in the literature. Jones~\cite{jones1973failure} presented one of the first work on failure detection in linear systems. They presented a linear filter that increases the sensitivity of the residual of the filter, which helps to improve the detection of a particular failure. Brumback et al.~\cite{brumback1987chi} presented a $\chi^2$ test for fault detection in Kalman filters. Mo et al.\cite{mo2010false} studied the effect of false data injection attacks on state estimation with a $\chi^2$ failure detectors. 

\ReviseTase{
In this paper, we study how to design spoofing signals that are agnostic to the failure detector. The problem is to minimize the magnitude of the injected signals while still ensuring the desired separation in the filter output. We also provide numerical simulations to show our strategy successfully misleads the $\chi^2$ detector. }

\ReviseTase{\textbf{Contributions}: In this paper, we make the following contributions:
\begin{itemize}
    \item \textbf{Problem:} We formalize the problem of designing false measurement data that is injected to corrupt and mislead the output of a  Kalman filter.  
    Intuitively, if the spoofing measurement is large enough, the spoofed estimation will deviate far away from the true state. However, it will be easy to trigger the alarm by potential failure detectors because of the abnormal estimation. Thus, we study how to design and inject the false measurement data to achieve the desired separation with minimal energy of injected data. This is the first work to formalize this problem.
    \item \textbf{Solution:}  We present an algorithm to solve this problem and show how to bypass residual-based failure detector: 
    \begin{itemize}
        \item \emph{False data design:} We derive the relationship between a spoofed estimation state and the normal state. We show the problem can be solved by Linear Programming optimally and efficiently. 
        \item \emph{Spoofing with failure detector:} We provide strategies and a sufficient condition to design false data for a Kalman filter with a residual-based failure detector (such as a $\chi^2$ detector) without being detected.
    \end{itemize}
    \item \textbf{Evaluation:} We demonstrate with MATLAB simulations that we can archive any desired separation accurately by injecting false measurement data. We also show that by gradually increasing the separation with a given condition, the attack strategy is able to successfully mislead the $\chi^2$ detector. 
\end{itemize}
}

Based on the motion model of the target and the evolution of the KF, three problems for spoofing design are formulated in Section \ref{sec:probform}. Section \ref{sec:Strategies} shows the approaches to solve these optimization problems. The simulations for verifying spoofing strategies are given in Section~\ref{sec:simulation}. Section~\ref{sec:detector} provides a numerical example to illustrate how the proposed spoofing strategy can be applied to a system equipped with a failure detector. Finally, Section~\ref{sec:conc} summarizes the conclusion and future work. 

\section{Problem Formulation} \label{sec:probform}
\noindent\textbf{Notation:} We denote the set of positive real number by $\mathbb{R}^{+}$, the set of positive integer by $\mathbb{Z}^{+}$. The set of real vectors with dimension $n$ is denoted by $\mathbb{R}^{n},~n\in \mathbb{Z}^{+}$, and the set of real matrices with $m$ rows and $n$ columns by $\mathbb{R}^{m\times n}, ~m, n\in \mathbb{Z}^{+}$. We write $\lVert \cdot \lVert_{p}^{p}, ~p\in\mathbb{Z}^{+}$ as the $p^{th}$ power of $L_{p}$ vector norm,  $\mathbb{E}(\cdot)$ as the expectation of a random variable, $I_{n}$ as the identity matrix with size $n, ~n\in \mathbb{Z}^{+}$, and $\mathcal{N}(\mu,\sigma^{2})$ as the normal distribution with mean $\mu$  and variance $\sigma^{2}$.

We consider a scenario where an observer estimates the location of a target using a  KF in an \Revised{$n$--dimensional space}. The target misleads the observer by adding spoofing signals to the observer's measurement. We define the target's  model as: 
\begin{equation}  \label{eqn:target_motion_model}
  x_{t+1}=\mathcal{F}x_{t}+\mathcal{G}u_t+\omega_t,
\end{equation}   
\Revised{where $\mathcal{F},\mathcal{G}\in\mathbb{R}^{n\times n}$, $x_{t}\in\mathbb{R}^n$ is the state of  target, $u_t\in \mathbb{R}^{n}$ is the  control input, $w_t\sim \mathcal{N}(0, R)$ is the Gaussian process noise with $R\in \mathbb{R}^{n\times n}$}. 

The observer estimates the target's measurement  using a linear measurement model:
\begin{equation}
  \label{eqn:realmeasurement}
   z_t= \mathcal{H}x_t +v_t,
\end{equation}
where \Revised{ $\mathcal{H}\in \mathbb{R}^{n\times n}$}  and $v_t \sim \mathcal{N}(0,Q)$ gives the measurement noise with \Revised{$Q\in \mathbb{R}^{n\times n}$}. \ReviseTase{Note that $R$ and $Q$ are both positive semi-definite matrices.}

In order to mislead the observer, the target corrupts the observer's measurement by adding spoofing signal to mislead the observer's estimate. We assume the measurement received by the observer is \Revised{$\tilde{z}_t\in \mathbb{R}^n$} with spoofing signal (Equation~\eqref{eqn:spoofmeasurement}) instead of the true measurement \Revised{$z_t\in \mathbb{R}^n$} without spoofing signal (Equation~\eqref{eqn:realmeasurement}). The spoofing signal \Revised{$\epsilon_t:=[\epsilon_{t1},\cdots \epsilon_{tn}]^T\in \mathbb{R}^n$} adds additional measurement error: 
\begin{equation}
  \label{eqn:spoofmeasurement}
 \tilde{z}_t= z_t + \epsilon_t.
\end{equation}

\begin{figure}
  \centering
  \includegraphics[width=0.65\columnwidth]{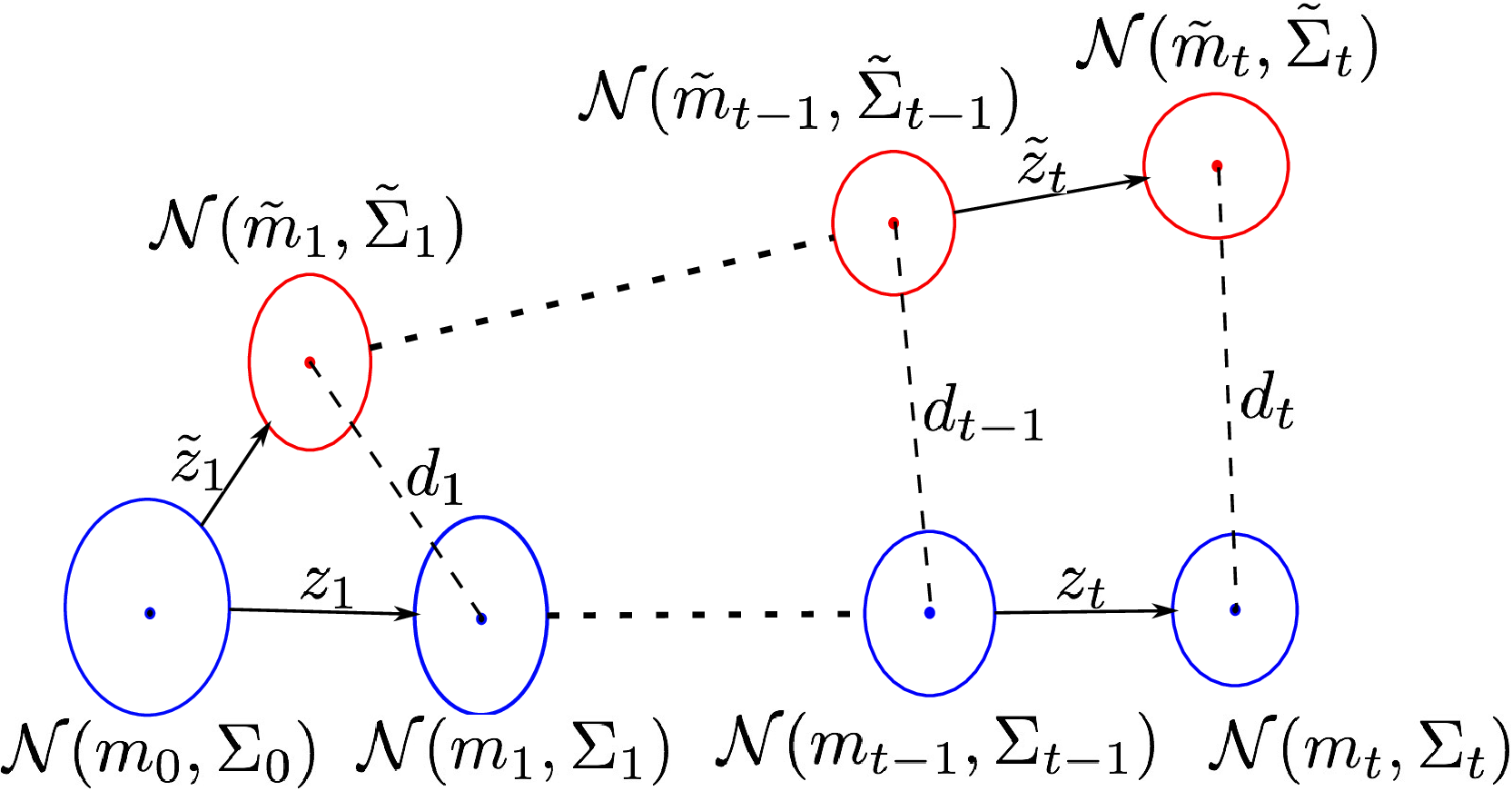}
  \caption{The evolution of KF estimate by applying $z_t$ and  $\tilde{z}_t$, respectively. \Revised{Note that $m_t$ and $\tilde{m}_t$ may also be different initially, i.e., ($\tilde{m}_0 \ne m_0$).}}
  \label{Distance}
\end{figure}

The observer uses a KF to estimate target's state with initial distribution $\mathcal{N}(m_0, \Sigma_0)$. Since it receives the spoofing measurement $\tilde{z}_t$ for updating, we denote distributions generated by the evolution of its KF as $\mathcal{N}(\tilde{m}_t, \tilde{\Sigma}_t)$ when step $t\geq 1, t\in \mathbb{Z}^{+}$. We also denote the distributions generated by the evolution of a KF using true measurement $z_t$ as $\mathcal{N}({m}_t, {\Sigma}_t)$. The goal for the target is to set the separation between the mean estimate $m_t$ and $\tilde{m}_t$. The target's spoofing signal is each step within the planning horizon for which some desired separation, $d_t\geq 0$, must be achieved (Figure~\ref{Distance}). Figure~\ref{explain_distance} shows the target's spoofing process where it uses the initial guess of $\mathcal{N}(m_0, \Sigma_0)$ denoted as $\mathcal{N}(\tilde{m}_0, \tilde{\Sigma}_0)$ and desired separation $d_t$ to design spoofing signal $\epsilon_t$. In order to avoid detection, the targets seeks to minimize the magnitude of the spoofing signal. 

\Revised{Note that, although we use the example of tracking a moving target, the state $x_t$ can be more general. For example, it can represent the state of a power system~\cite{liu2011false}, the state of a networked system\cite{mo2010false}, or the state of a GPS device~\cite{larcom2013modeling}.}
\begin{figure}
  \centering
  \includegraphics[width=0.75\columnwidth]{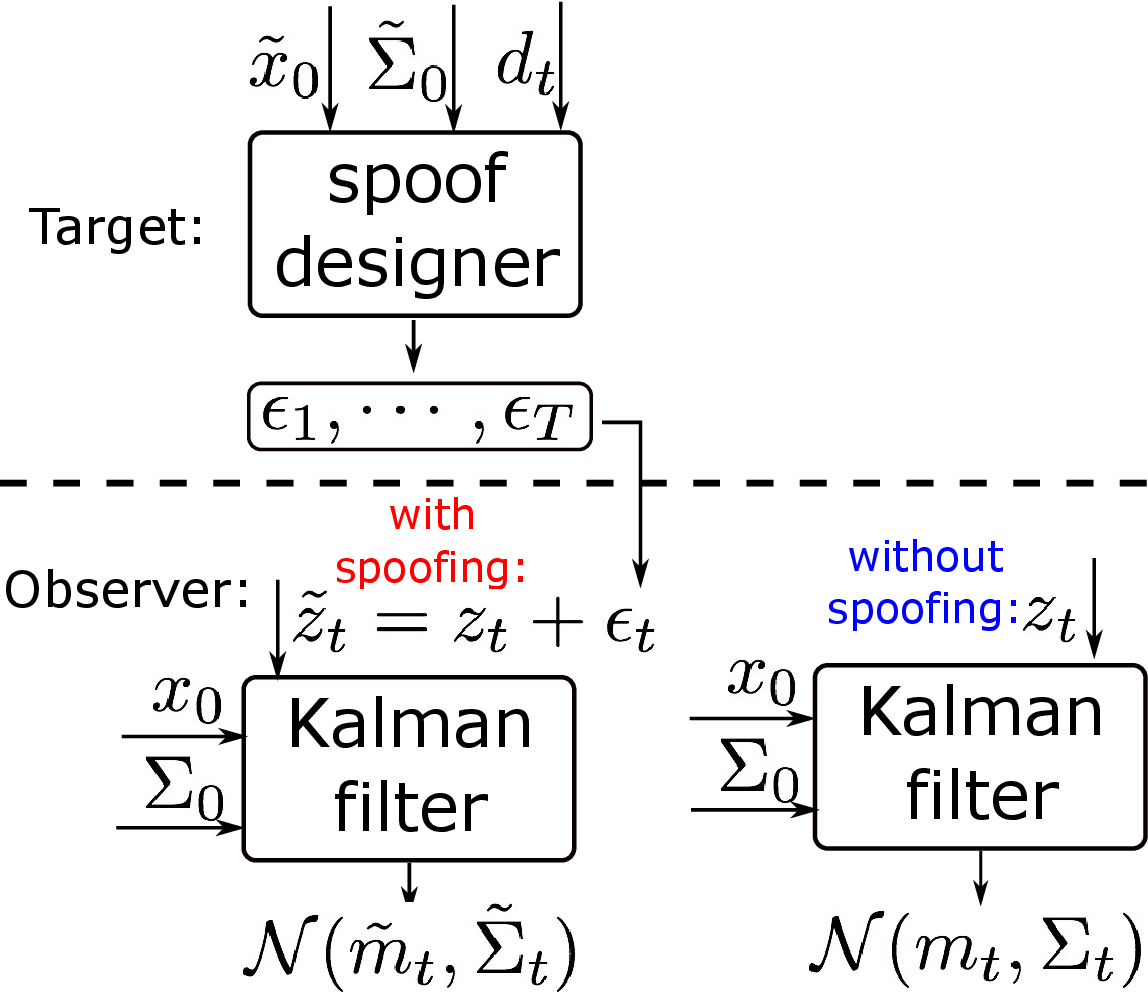}
  \caption{Signal spoofing process and its effect on the observer's KF estimation.}
  \label{explain_distance}
\end{figure}

We first propose two problems for offline scenarios as follows.
\subsection{Offline Spoofing Signal Design with Known  $\mathcal{N}(m_0, \Sigma_0)$}
If the target knows $\mathcal{N}(m_0, \Sigma_0)$ of the KF, then the target can set $\mathcal{N}(\tilde{m}_0, \tilde{\Sigma}_0)$ equal to $\mathcal{N}(m_0, \Sigma_0)$.  

\begin{problem} [Offline with Known $\mathcal{N}(m_0, \Sigma_0)$] \label{pro:problem1}
Consider a target with motion model (Equation~\eqref{eqn:target_motion_model}), measurement model (Equation~\eqref{eqn:realmeasurement}), and spoofing measurement model (Equation~\eqref{eqn:spoofmeasurement}). Assume target knows $\mathcal{N}(m_0, \Sigma_0)$. Find a sequence of spoofing signal inputs, $\{\epsilon_1, \epsilon_2,\cdots,  \epsilon_{T} \}$ to achieve desired separation $d_t$  between $\tilde{m}_t$ and $m_t$ at step $t$. Such that,
\begin{equation*}
  \label{eqn:obj_pro_1}
 \text{minimize}  \quad \sum^{T}_{t=1} \gamma_t \cdot \lVert \epsilon_t\lVert_p^p
\end{equation*}
subject to,
\begin{equation} \label{pro:problem1stconstraint}
\begin{split}
\quad \lVert m_t -& \tilde{m}_t\lVert_{p}^{p} \ge d_{t}^{p},\quad \forall t
      \end{split}
\end{equation}
where $\gamma_t\in\mathbb{R}^{+}$ is a weighing parameter and $T\in \mathbb{Z}^{+}$ is the optimization horizon. 

\end{problem}
\subsection{Offline Spoofing Signal Design with Unknown  $\mathcal{N}(m_0, \Sigma_0)$}
Next, we consider the case where the target does not know the initial condition in the KF. Instead, we assume that the initial estimate $\tilde{m}_0$ is not too far away from $m_0$ (in exception). 
\begin{problem}  [Offline with Unknown $\mathcal{N}(m_0, \Sigma_0)$] \label{pro:problem2}
Consider a target with motion model (Equation~\eqref{eqn:target_motion_model}), measurement model (Equation~\eqref{eqn:realmeasurement}), and spoofing measurement model (Equation~\eqref{eqn:spoofmeasurement}). Assume the target starts spoofing with $\tilde{m}_0$, where $ \mathbb{E}( m_0 - \tilde{m}_0) = M_0$  and $\tilde{\Sigma}_0 \neq {\Sigma}_0$. Find a sequence of spoofing signal inputs, $\{\epsilon_1, \epsilon_2,\cdots,  \epsilon_{T} \}$ to achieve desired separation $d_t$  between $\tilde{m}_t$ and $m_t$ (in expectation) at step $t$. Such that
\begin{equation*}
  \label{eq5}
 \text{minimize}   \quad \sum^{T}_{t=1} \gamma_t \cdot \lVert \epsilon_t\lVert_p^p 
\end{equation*}
subject to,
\begin{equation} \label{problem2st} 
\begin{split}
\quad \lVert \mathbb{E}(m_t -& \tilde{m}_t)\lVert_{p}^{p} \ge d_{t}^{p}, \quad \forall t
\end{split}
\end{equation}
where $\gamma_t\in\mathbb{R}^{+}$ is a weighing parameters and $T\in\mathbb{Z}^{+}$ is the optimization horizon.
\end{problem}
\section{Signal Spoofing Strategies} \label{sec:Strategies}
In this section, we show how to solve Problems~\ref{pro:problem1} and~\ref{pro:problem2} when $p=1$ and $p=2$. We first present the relationship between the separation $m_t - \tilde{m}_t$  and the initial bias $m_0 - \tilde{m}_0$.
\begin{theorem} \label{mk} 
Consider a target with motion model (Equation~\eqref{eqn:target_motion_model}), measurement model (Equation~\eqref{eqn:realmeasurement}), and spoofing measurement model  (Equation~\eqref{eqn:spoofmeasurement}). The evolutions of the KFs by applying $z_t$ and $\tilde{z}_t$ give the distributions $\mathcal{N}(m_t, \Sigma_t)$ and $\mathcal{N}(\tilde{m}_t, \tilde{\Sigma_t})$, respectively. The difference, $m_t  - \tilde{m}_t$ is,
\begin{equation} \label{mt}
\begin{split}
m_t  - \tilde{m}_t = &\prod_{i=0}^{t-1}A_{t-i}\cdot (m_0  - \tilde{m}_0) + \\
&\sum_{i=0}^{t-2} \left( \prod_{j=0}^{i}A_{t-j}  (B_{t-1-i} + C_{t-1-i}) \right) +B_t + C_t,
\end{split}
\end{equation}
where $A_t = \mathcal{F} - \tilde{K}_t \mathcal{H}\mathcal{F},\quad B_t = (K_t-\tilde{K}_t)\left[z_t-\mathcal{H}(\mathcal{F}m_{t-1}+\mathcal{G}u_{t-1})\right], \quad
C_t= - \tilde{K}_t \epsilon_t$. 
\end{theorem}

The proof is given in the appendix. 

\begin{corollary} \label{corr:expected}
The expected value of the separation is, 
\begin{equation}
\begin{split} \label{mt_expectation}
&  \mathbb{E}\left(m_t-\tilde{m}_t\right)  \\
&=  \prod_{i=0}^{t-1}A_{t-i} M_0 +\sum_{i=0}^{t-2} \left(\prod_{j=0}^{i}A_{t-j} C_{t-1-i}\right) +  C_t.
\end{split}
\end{equation}
\end{corollary}
\begin{proof}
From Equation~\ref{mt}, $\mathbb{E}\left(m_t-\tilde{m}_t\right)$ follows,
\begin{equation*} 
\begin{split}
&  \mathbb{E}\left(m_t-\tilde{m}_t\right)   \\
&=  \mathbb{E} \left(\sum_{i=0}^{t-2}  \prod_{j=0}^{i}A_{t-j}\cdot B_{t-1-i}    +  B_t \right) + \\
&\quad\prod_{i=0}^{t-1}A_{t-i} \mathbb{E} (m_0  - \tilde{m}_0) +\sum_{i=0}^{t-2} \left( \prod_{j=0}^{i}A_{t-j} C_{t-1-i}\right) +  \tilde{K}_t\epsilon_t.
\end{split}
\end{equation*}
The actual measurement is: $z_i = \mathcal{H}(\mathcal{F}m_{i-1} +\mathcal{G}u_{i-1}+w_i)+v_i$, where $w_i$ and $v_i$ are Gaussian noises with zero mean. The expected measurement value is: $\mathbb{E}(z_i)= \mathcal{H}( \mathcal{F}m_{i-1}+ \mathcal{G}u_{i-1})$ for all $i$, thus $\mathbb{E}[z_i- \mathcal{H}( \mathcal{F}m_{i-1}+ \mathcal{G}u_{i-1})]=0$.  Since $\mathbb{E}[B_i] = 0$, we have, 
\begin{equation}
\begin{split} 
&  \mathbb{E}\left(m_t-\tilde{m}_t\right)  \\
&=  \prod_{i=0}^{t-1}A_{t-i}\mathbb{E} (m_0  - \tilde{m}_0) +\sum_{i=0}^{t-2} \left(\prod_{j=0}^{i}A_{t-j} \tilde{K}_{t-1-i} \epsilon_{t-1-i}\right)\\
&\quad +  \tilde{K}_t\epsilon_t.
\end{split}
\end{equation}
Since we assume $ \mathbb{E}( m_0 - \tilde{m}_0) = M_0$ in Problem~\ref{pro:problem2}, the claim is guaranteed.
\end{proof}


Theorem \ref{mk} shows the difference between the two estimated means at step $t$ depends on the initial means, $m_0$ and $\tilde{m}_0$, and the initial covariance matrices $\Sigma_0$ and $\tilde{\Sigma}_0$. This is because the Kalman gain $K_t$ depends on the covariance matrix $\Sigma_t$. If target sets $m_0=\tilde{m}_0$ and $\Sigma_0 = \tilde{\Sigma}_0$, it has $\Sigma_t = \tilde{\Sigma}_t$ for all $t$ since the covariance matrix is updated through the same Kalman prediction and update equation (see appendix). Thus, $B_t = 0_{2\times 2}$ and then Equation \eqref{mt} can be simplified as:
$$  m_t  - \tilde{m}_t =\sum_{i=0}^{t-2} \left(\prod_{j=0}^{i}A_{t-j}  C_{t-1-i}\right) + C_t. $$
As a result, $m_t  - \tilde{m}_t $ is independent of the measurements $\{z_1, z_2, \cdots, z_t\}$ when $m_0=\tilde{m}_0$ and $\Sigma_0 = \tilde{\Sigma}_0$. Thus, the target can generate spoofing signal inputs by solving Problem \ref{pro:problem1} offline. Similarly, following Corollary~\ref{corr:expected}, Problem~\ref{pro:problem2} can be saved offline as well.

 Problems~\ref{pro:problem1} and~\ref{pro:problem2} are two nonlinear programming problems for arbitrary vector norms $L_p$. However, 
when $p=1$, they can be formulated as linear programming problems. Linear programming can be solved in polynomial time \cite{karmarkar1984new}. When $p=2$, they become QCQP (Quadratically Constrained Quadratic Program). The following shows the LP and QCQP formulations.
\begin{theorem}
If $p=1$ and the elements in $ \mathcal{F}$ and $I-K_t \mathcal{H}$ are all positive, then Problems~\ref{pro:problem1} and \ref{pro:problem2} can be solved optimally with linear programming. If $p=1$ and the elements in $ \mathcal{F}$ and $I-K_t \mathcal{H}$ are not all positive, then Problems~\ref{pro:problem1} and \ref{pro:problem2} can be solved optimally with $4^k$ linear programming instances. If $p=2$ and $\{ \mathcal{H}, \mathcal{F},Q,R\}$ are diagonal matrices, then Problems~\ref{pro:problem1} and \ref{pro:problem2} can be solved optimally with linear programming. 
\end{theorem}

\subsection{Linear Programming Formulation  for $L_1$ Vector Norm}
\label{sec:$L_1$ norm}
Here, we show how to formulate Problem~\ref{pro:problem1} using linear programming. A similar procedure can be applied to formulate Problem~\ref{pro:problem2} as linear programming.

The constraint in Problem~\ref{pro:problem1} (Equation~\ref{pro:problem1stconstraint}) follows:
\begin{align}\label{replacePro1}
&\lVert m_t-\tilde{m}_t\lVert_1  =\left\lVert\sum_{i=0}^{t-2} \left(\prod_{j=0}^{i}A_{t-j}  C_{t-1-i}\right) + C_t \right\lVert_1\nonumber\\
&= \left\lVert \sum_{i=0}^{t-2} \left( \prod_{j=0}^{i}A_{j+1}\cdot \tilde{K}_{t-1-i}  \cdot \epsilon_{t-1-i}\right) + \tilde{K}_t \epsilon_t \right\lVert_1\nonumber\\
&\geq d_t,
\end{align}
where $t=1,2,\cdots,T$. $\prod_{j=0}^{i}A_{t-j}\cdot \tilde{K}_i\in \mathbb{R}^{2\times 2}$ is a constant matrix for each $i\in\{1,\cdots,t-1\}$ and is calculated from the KF iteration with initial covariance $\Sigma_0$ and $\tilde{\Sigma}_0$. 
Since $L_{1}$ vector norm is the sum of the absolute values of the elements for a given vector, Problem~\ref{pro:problem1} can be directly formulated as a linear programming problem when $p=1$. 

Then we show how to transform this constraint to a standard linear constraint form $G_t x_t \ge d_t$. \Revised{To simplify the equation, we use a 2-D case as an example}, with $x_t:=[\epsilon_{1x},\cdots,\epsilon_{tx}, \epsilon_{1y},\cdots,\epsilon_{ty}]^{T}$. The problem can be extended to $n$ dimension follow the same idea.  The left side of Equation \eqref{replacePro1} can be formulated as
\begin{equation}
\begin{split}
\lVert m_t - \tilde{m}_t\lVert_1=
\bigg\lVert                 
  \begin{array}{cc} 
   a_0+ a_{1} \epsilon_{1x} +\cdots a_{t}\epsilon_{tx} +\cdots+ a_{2t}\epsilon_{ty}  \\
   b_0+ b_{1} \epsilon_{1x} +\cdots b_{t}\epsilon_{tx} + \cdots + b_{2t}\epsilon_{ty} \\
  \end{array}
\bigg\lVert_1 
\label{eqn:specific_constraint_example}
\end{split}
\end{equation}
where $a_0,a_{1},\cdots,a_{2t},b_0,b_{1},\cdots,b_{2t}$ are corresponding coefficients from Equation~\ref{mt}.
\begin{lemma}
\label{lemma_linear}
If the elements in matrices $ \mathcal{F}$ and $I - K_t \mathcal{H}$ are positive, then  $\lVert m_t - \tilde{m}_t\lVert_1$ is a linear combination of $|\epsilon_{ix} |$ and $|\epsilon_{iy}|$, and Problem~\ref{pro:problem1} can be solved as a single LP instance.
\end{lemma}
\begin{proof}
According to the proof of Theorem~\ref{mk} appendix, all the coefficients $\{a_1,...,a_{2t}, b_1,...,b_{2t}\}$ are positive if the elements in matrices $ \mathcal{F}$ and $I - K_t \mathcal{H}$ are positive. Therefore, the objective function and the constraints are linear in $|\epsilon_{ix} |$ and $|\epsilon_{iy}|$. There always exists an optimal solution where all $\epsilon_{ix} \geq 0$ and $\epsilon_{iy}\geq 0$ or where all $\epsilon_{ix} \leq 0$ and $\epsilon_{iy}\leq 0$. The objective function in both cases will be the same. Without loss of generality, we can assume $\epsilon_{ix} \geq 0$ and $\epsilon_{iy}\geq 0$, which can be solved using a single LP instance.
\end{proof}

The linear programming strategy  containing $k$ constraints is presented in Algorithm~\ref{alg:linear_programming}. $G$ denotes matrix in the linear constraint $G x \geq D_k$ where $x:=[\epsilon_{1x},\cdots,\epsilon_{Tx}, \epsilon_{1y},\cdots,\epsilon_{Ty}]^{T}$ and $D_k$ is the collection of $k$ nonzero separations $d_t, ~t\in\{1,\cdots, T\}$.
\begin{algorithm}\label{alg:linear_programming}
    \SetAlgoLined
$\mathbf{Initial}\leftarrow\left\{(x_o, \Sigma_0, \mathcal{F},\mathcal{H},  \mathcal{G}, Q, R, u \right\}$\\  
$G\leftarrow 0_{k\times n\cdot T}$ \\
Calculate Kalman gain $\tilde{K}_1,\cdots,\tilde{K}_T$ \\
\For{$i=1$ to the $q_{th}$ value in $D_k$}{
$g = \prod_{j=i}^{T-1}A_{j+1} \tilde{K}_i$\label{algline:g}\;
$G_{q,i} = \text{sum  of  all  rows in}$ $g$
}
$\mathbf{Return} \quad G$
    \caption{Linear Programming Formulation}
    \label{minimax} 
\end{algorithm}

If Lemma \ref{lemma_linear} does not  hold, it is possible that some elements in $a_0,a_{1},\cdots,a_{2t},b_0,b_{1},\cdots,b_{2t}$ can be positive and some are negative. In general,  there are four different cases depending on the sign of the first row and the second row for considering each constraint $\lVert m_t - \tilde{m}_t\lVert_1 \geq d_t$ (Equation~\ref{eqn:specific_constraint_example}). Then we can obtain four linear optimization problems along four different sub-constraints of each constraint $\lVert m_t - \tilde{m}_t\lVert_1 \geq d_t$. Thus, in the worst case, the optimal solution can be obtained by solving $4^k$ linear optimization problems. We run Algorithm~\ref{alg:linear_programming} $4^k$ times by changing the sign of rows in $g$ (Line~\ref{algline:g}) appropriately.

\subsection{Quadratically Constrained Quadratic Program Formulation for $L_2$ Vector Norm}
When $p=2$, Problems~\ref{pro:problem1} and~\ref{pro:problem2} can be formulated as  QCQPs\cite{boyd2004convex}:

\begin{equation}
\begin{split}
&\text{minimize} \quad  \frac{1}{2}x_{\epsilon}^{T}P_0 x_\epsilon \\
&\text{s.t.}\quad -\frac{1}{2} x^T_\epsilon D^T_t D_t x_\epsilon + d_{t}^{2} \le 0 ,\quad\forall t\in\{1,\dots,T\}\\
\label{eqn:QCQP_forthree}
\end{split}
\end{equation}

For simplify the equation, we use a 2-D case as the example, where $x_{\epsilon}=[\epsilon^2_{1x},\epsilon^2_{1y},\cdots,\epsilon^2_{Tx},\epsilon^2_{Ty}]^{T}$, $P_0=I_{2T}$, and 

$D_t\in \mathbb{R}^{2T \times 2T} :=$ 
\[
  \begin{bmatrix}
    \prod_{j=1}^{t-1}A_{j+1} \tilde{K}_0 & &\cdots &0 & 0 &0\\
    \vdots& \ddots &\vdots &\vdots &\vdots &\vdots\\
    0& \cdots& \prod_{j=t-1}^{t-1}A_{j+1} \tilde{K}_{t-1} &0 &0 &0\\
    0& \cdots& 0&  \tilde{K}_t & 0 & 0\\
    0& \cdots& 0&  0 & 0 & 0\\
    0& \cdots& 0&  0 & 0 & \ddots
  \end{bmatrix}
\]

Unfortunately, the QCQP formulations for these three problems are NP-hard since the constraint in each problem is concave. If $ \mathcal{F}, \mathcal{G}, \mathcal{H},\tilde{\Sigma_0}$ are diagonal matrices, it can be shown that $D_t$ is also a diagonal matrix. We can transform the QCQP formulation to a linear programming problem by using change of variables $\{\epsilon^2_{tx},\epsilon^2_{ty}\}, ~t=\{1,2,...,T\}$, and using  a procedure similar to $p=1$.

If $D_t$ is not a diagonal matrix, one solution is to apply the inequality $\sqrt[]{2} \lVert x \lVert_2 \ge  \lVert x \lVert_1$ between $L_1$ vector norm and $L_2$ vector norm. The constraint can be changed to  $L_1$ vector norm, which is a stricter constraint. A  sub-optimal solution can be obtained by using the $L_1$ vector norm. 
\subsection{Receding Horizon: Spoofing with online measurement}\label{pro:problem3}
 Problems \ref{pro:problem1} and \ref{pro:problem2} describe the offline versions for spoofing. We also extend the offline problems to an online version. The following formulates an online spoofing  scenario.
 
Consider a target with motion model (Equation~\eqref{eqn:target_motion_model}), measurement model (Equation~\eqref{eqn:realmeasurement}), and spoofing measurement model (Equation~\eqref{eqn:spoofmeasurement}). Assume the target does not know $\mathcal{N}({x}_0, {\Sigma}_0)$. It collects a series of measurements $\{z^{real}_1, z^{real}_2, \cdots, z^{real}_{t^o} \}$ from step $1$ to current step $t^{o}$. Find a sequence of spoofing signal inputs, $\{\epsilon_{t^{o}}, \epsilon_{t^{o}+1},\cdots,  \epsilon_{t^{o}+H} \}$ to achieve desired separation $d_t$  between $\tilde{m}_t$ and $m_t$ (in expectation) within future $H$ steps. Such that
\begin{equation*}
  \label{eq5}
  \text{minimize}   \quad \sum^{t^o+H}_{t=t^{0}} \gamma_t \cdot \lVert \epsilon_t\lVert_p^p
\end{equation*}
\begin{equation} \label{pro:pro3_constrain}
\begin{split}
\text{s.t.} \quad \lVert \mathbb{E}(m_t -& \tilde{m}_t)\lVert_{p}^{p} \ge d_{t}^{p} ,\quad \forall t\in\{t^o,\cdots,t^o+H\}
\end{split}
\end{equation}
where $\gamma_t\in\mathbb{R}^{+}$ is a weighing parameter, $t^o$ is the current time, and $H$ is the predictive time horizon. The target applies $\epsilon_{t}=\epsilon_{t^o}$ as spoofing signal input at each step $t$.

\section{Simulations} \label{sec:simulation}
In this section, we simulate the effectiveness of spoofing strategies for Problems~\ref{pro:problem1},~\ref{pro:problem2} and online case (Section~\ref{pro:problem3}) where a target designs spoofing signals $\epsilon_t$ to  mislead an observer by achieving the desired separations $d_t$ between $m_t$ and $\tilde{m}_t$. Our code  is available online.\footnote{\url{https://github.com/raaslab/signal_spoofing.git}}

We consider the $L_1$  vector norm and the following models,
$$ \mathcal{F}=I_{2\times2},  \mathcal{G}=I_{2\times2}, u= \left[\begin{matrix}
  1  \\
   1  \\
  \end{matrix} \right], 
  R =0.5I_{2\times2}, 
    Q =0.5I_{2\times2}.\\$$
Set the weight $\gamma_t=1$ for all $t$. 
 
For Problem~\ref{pro:problem1}, set the initial condition for the KF as,  
$$\Sigma_0 =I_{2\times2},
     ~m_0=\left[\begin{matrix}
  0  \quad  0  
    \end{matrix}\right]^T.
     $$
Since the target knows $\mathcal{N}({x}_0, {\Sigma}_0)$, it sets $\tilde{m}_0=m_0$ and $\tilde{\Sigma}_0 = \Sigma_0$. We first consider a scenario where the target wants to achieve the desired separation at steps, $t=5, 10, 15$, denoted as $d_5=1.77$, $d_{10}=3.54$ and $d_{15}=5.30$ with the optimization horizon $T=20$. The target generates  a sequence of spoofing signals $\{\epsilon_1,\cdots,\epsilon_{20}\}$ offline by using a linear programming solver. The spoofing performance is shown in Figure~\ref{Simulation_problem1}-(a) where the true separations are the same as the desired separations. Same successful spoofing achieved when the desired separations are chosen as $d_t = 0.25 \sqrt{2} t,~t=\{3,...,15\}$, as shown in Figure~\ref{Simulation_problem1}-(b). 

\Revised{The problem formulation applies  in higher dimensional systems as well, not just 2D. Figure~\ref{3Dexample} shows an example of misleading a KF in a 3D environment.}

\begin{figure} 
\centering{
\subfigure[Desired separations, $d_5 = 1.77$ and $d_{10}=3.54$, with $T=20$.]{\includegraphics[width=0.8\columnwidth]{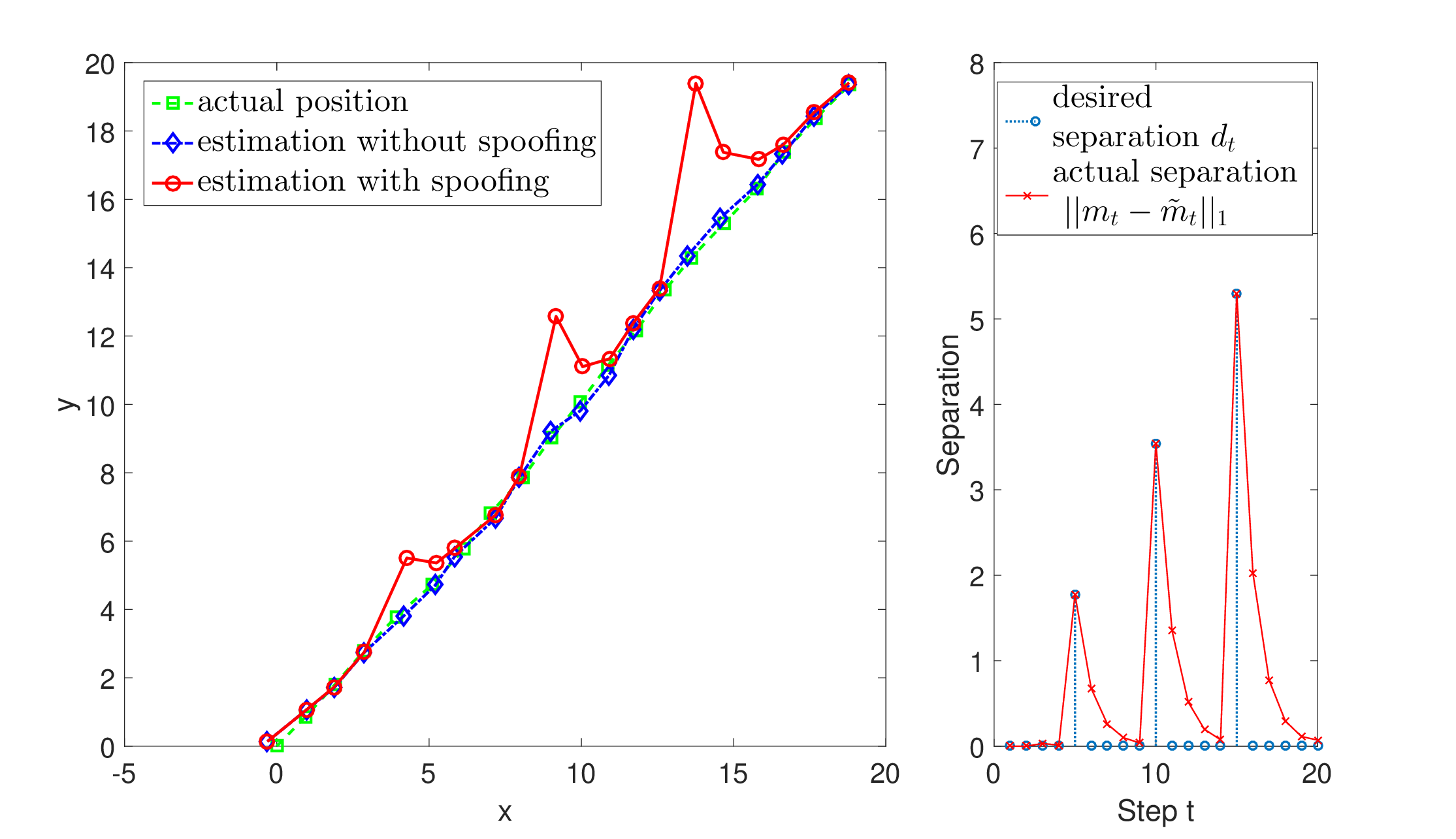}}
\subfigure[Desired separations, $d_t = 0.25 \sqrt{2}t$, with $t=3$ to $T=15$.]{\includegraphics[width=0.8\columnwidth]{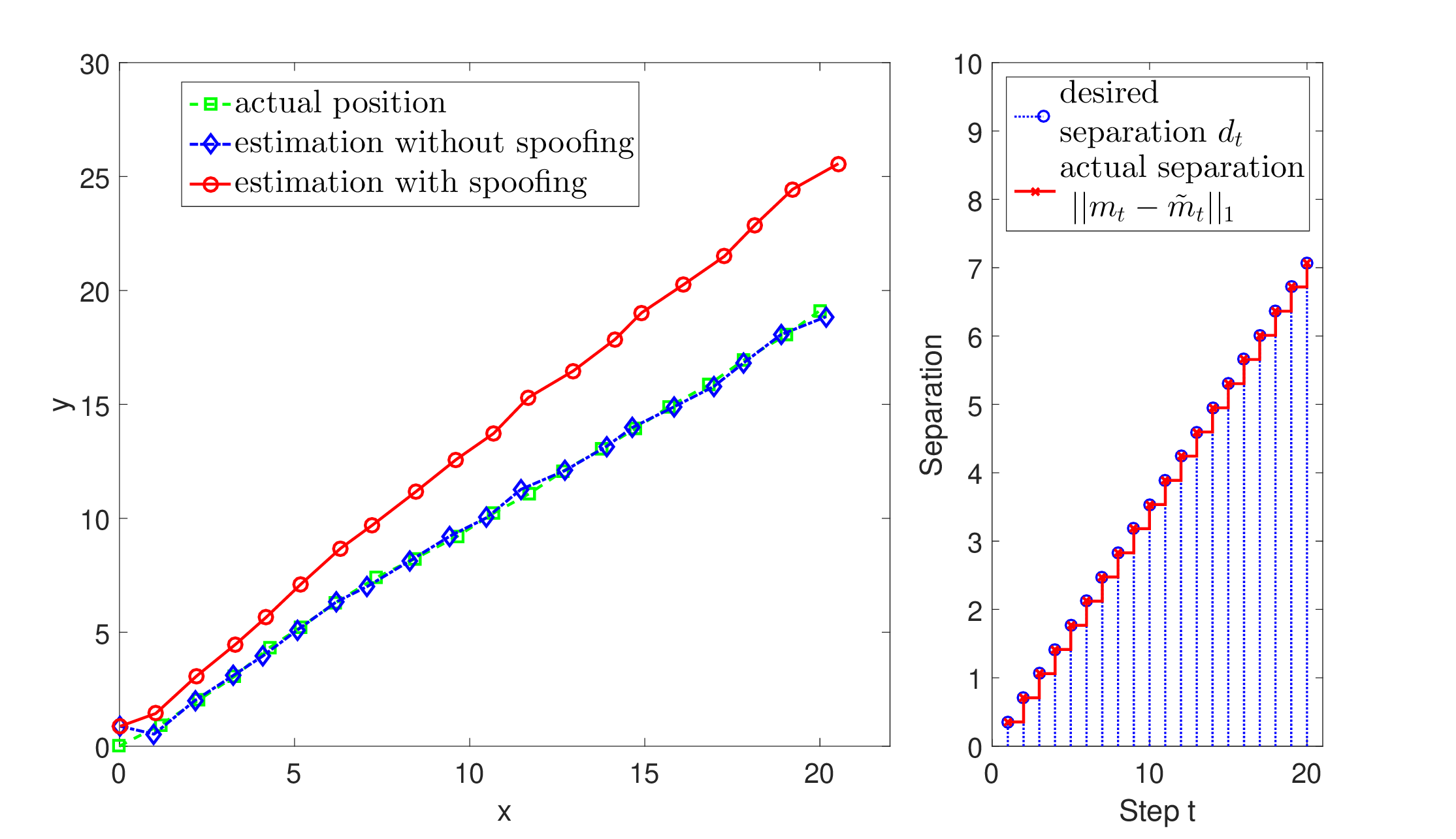}}
}
\caption{Offline signal spoofing with known ($m_0, \Sigma_0$).}
\label{Simulation_problem1}         
\end{figure}

\begin{figure} 
\centering
\includegraphics[width=0.8\columnwidth]{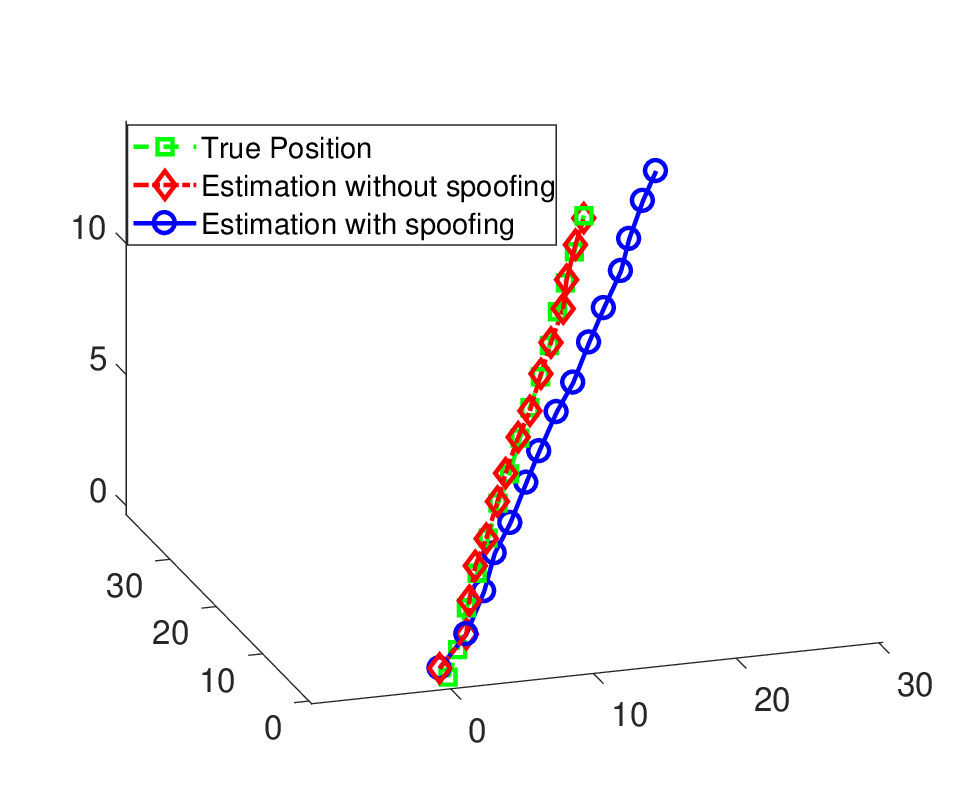}
\caption{\Revised{Signal spoofing with known ($m_0, \Sigma_0$) in 3D environment. Desired separations, $d_3 = 3, d_3 = 4,\cdots,d_{10} = 10$.} }
\label{3Dexample}             
\end{figure}

In Problem \ref{pro:problem2}, the target knows $ \mathbb{E}(m_0 - \tilde{m}_0) = M_0$ but does not know $\Sigma_0$. The spoofing result is no longer deterministic but holds in expectation $\lVert\mathbb{E}( m_t - \tilde{m}_t)\lVert_1 \ge d_t$. Figure~\ref{Simulation_problem2}-(a) shows spoofing signals for desired separations as $d_1=2$ with  $T=6$ and $M_0=1$. Set $\mathcal{N}(\tilde{m}_0,\tilde{\Sigma}_0)$ as $\mathcal{N}(0,1.5I_{2})$, $m_0$ as a random variable ($m_0 \sim \mathcal{N}(1,1)$) and ${\Sigma}_0 = I_{2}$. In order to see the effectives of the spoofing signals $\{\epsilon_1,\cdots,\epsilon_5\}$, we conduct 100 trials for each desired separation $d_2\in\{1,2,3,4,5\}$. Figure~\ref{Simulation_problem2}-(b) shows the $\lVert m_1 - \tilde{m}_1 \lVert_1$ is no longer deterministic, but  $\lVert \mathbb{E}(m_1 - \tilde{m}_1) \lVert_1 $ is  close to the desired value $d_1=2$. 
\begin{figure} 
\centering
\subfigure[Desired separation, $d_1 = 2$.]{\includegraphics[width=0.8\columnwidth]{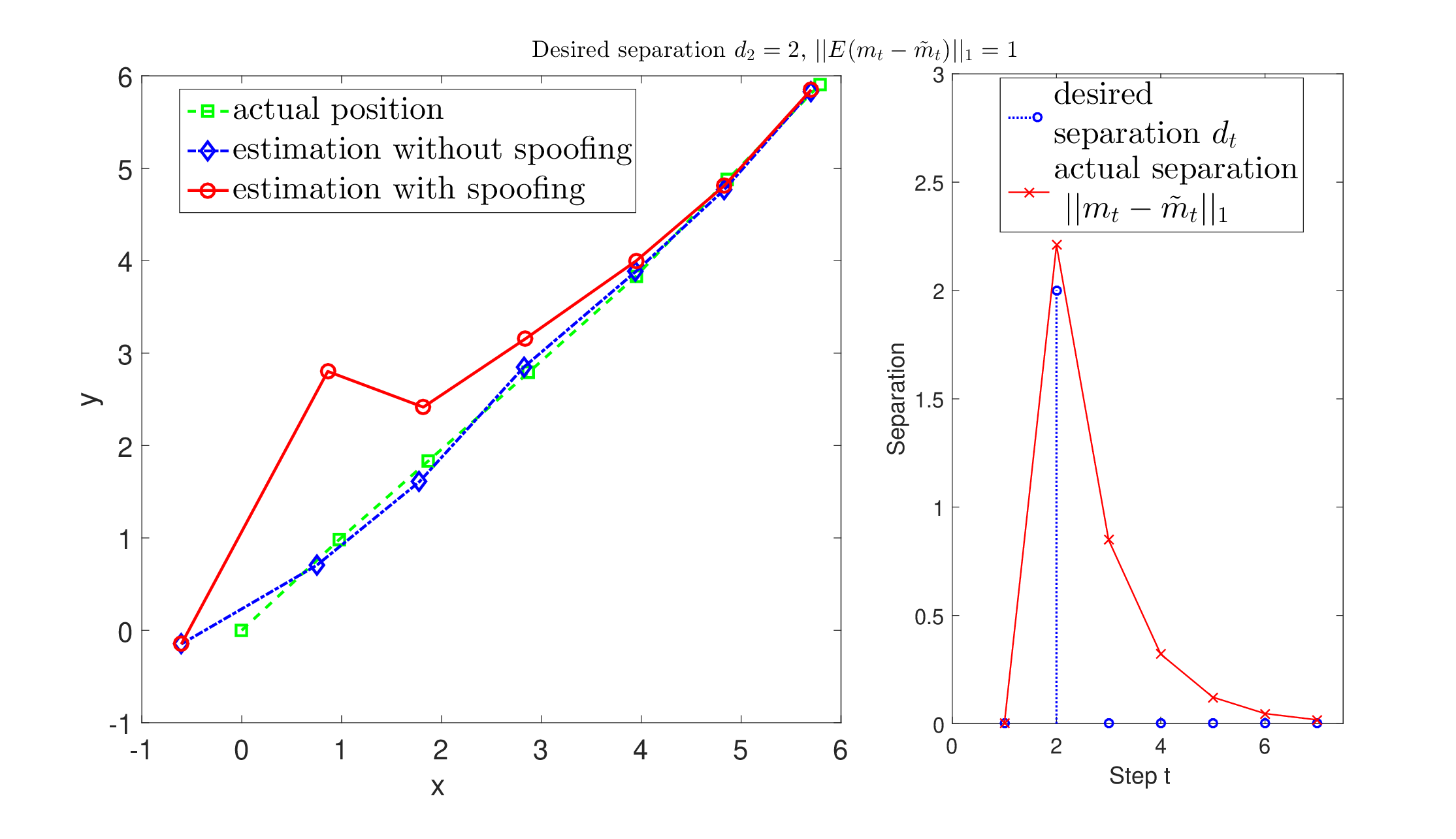}}
\subfigure[Results with $d_1$ = $\{1,2,3,4,5\}$ for 100 trials.]{\includegraphics[width=0.65\columnwidth]{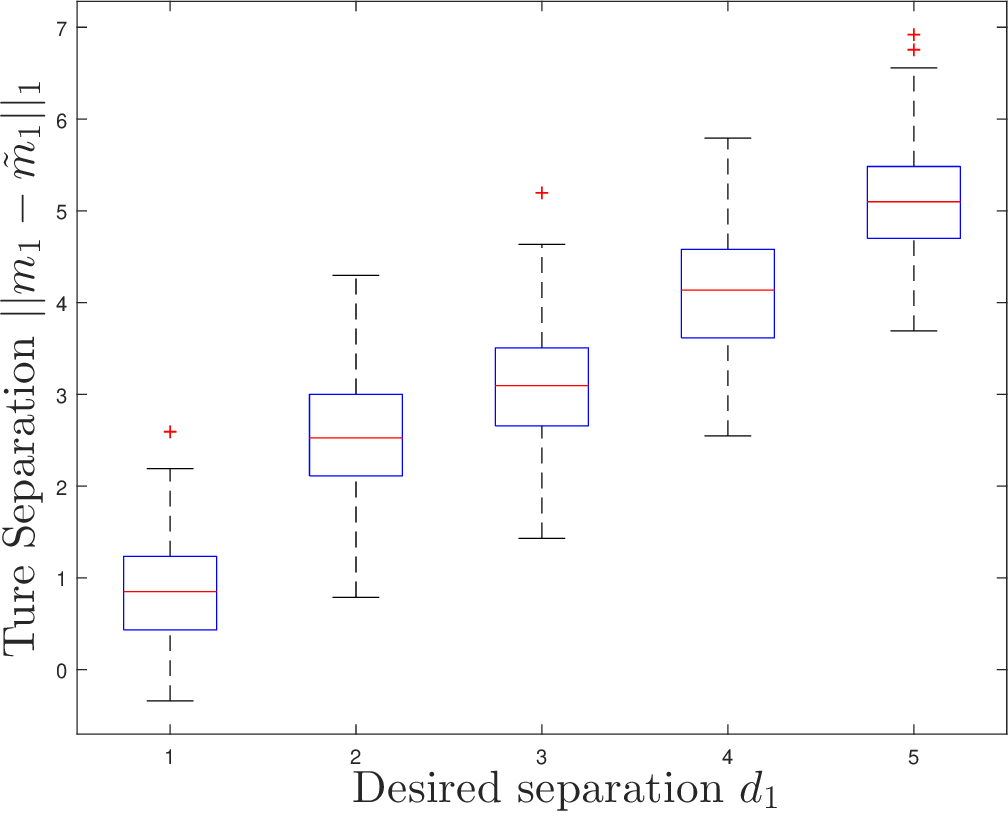}}
\caption{Offline signal spoofing with unknown ($m_0, \Sigma_0$).}
\label{Simulation_problem2}             
\end{figure}

For online case, spoofing signals are continuously generated by using receding horizon optimization with new noisy measurements. We set the receding horizon as $H=15$. Even though offline strategy performs comparatively as online strategy (Figure~\ref{Simulation_online}), online spoofing strategy achieves almost the same separation as the desired, while offline strategy has certain divergence. This is because online strategy can update the measurement at each step.

\begin{figure} 
\centering
{\includegraphics[width=0.65\columnwidth]{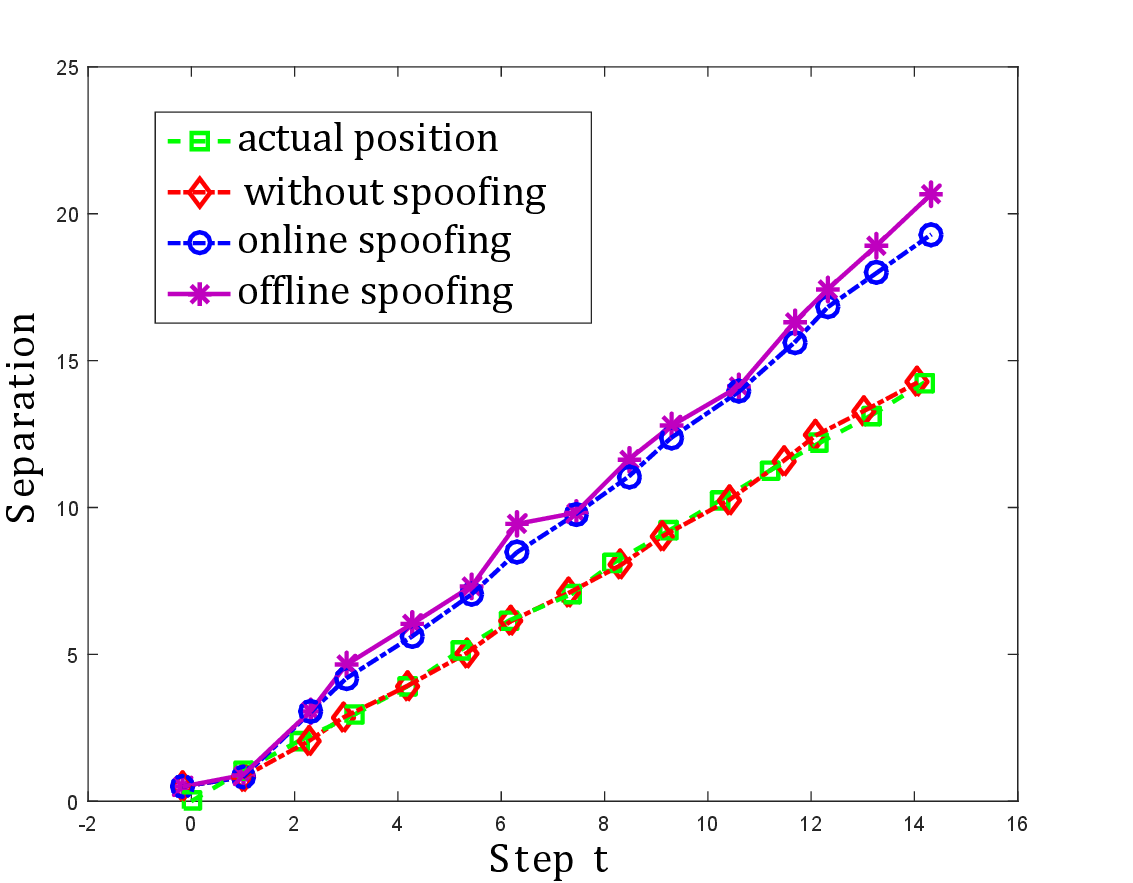}}
\caption{Online spoofing and offline spoofing with unknown ($m_0, \Sigma_0$).}
\label{Simulation_online}           
\end{figure}

\section{Signal spoofing with failure detector} \label{sec:detector}
\Revised{In this section, we evaluate the performance of the false data injection strategy in the presence of a failure detector. We show the conditions (Theorem~\ref{the:detector}) under which the generated false data can mislead a $\chi^2$ detector. This result can be also extended to other residual-based detectors.} 
\Revised{
\subsection{$\chi^2$ failure detector}
A $\chi^2$ detector computes the following measure,
\begin{equation}
g_t = r_t^T \Sigma_{r_t}^{-1} r_t,
\label{eqn:chi2}
\end{equation}
where $r_t = z_t - \mathcal{H}m_t$ is the innovation or measurement residual of the KF. Here, $\Sigma_{r_t}$ is the covariance matrix of the residual~\cite{thrun2005probabilistic}. The residual is Gaussian since it is the linear combination of two Gaussian random variables. It is known that $g_t$ is $\chi^2$ distributed with $n$ degrees of freedom. If $g_t>\text{threshold}$, the detector raises an alarm that the filter is under attack~\cite{brumback1987chi}.}

\Revised{
First, we review the Kalman Filter update equations, 
\begin{align}
&{m}_{t|t-1} =  \mathcal{F}m_{t-1|t-1} +  \mathcal{G}u_t,\label{eqn:kalman1} \\
&{m}_{t|t}= \mathcal{F}{m}_{t|t-1}+K_t(z_t- \mathcal{H}{m}_{t|t-1}), \label{eqn:kalman2}
\end{align}
where $K_t$ is the Kalman gain and is given by,
\begin{equation}
\label{eqn:kalman3}
K_t = (\mathcal{F} {\Sigma}_{t|t-1} \mathcal{F}' + R_t)\mathcal{H}'(\mathcal{H}\Sigma_{t|t-1} \mathcal{H}' + Q_t )^{-1}.
\end{equation}
We use the notation, $\tilde{\cdot}$, to indicate the system under attack.
}

\Revised{
Intuitively, lower the amount of injected attack signal, the less likely it will be detected. This is the motivation behind reducing the energy of the injected system. Nevertheless, when designing an attack sequence over a time horizon, we may have to carefully design the separation sequence $d_1,d_2,\cdots,$ so that they are not too large. In the following, we modify notion of a successful attack from~\cite{mo2010control} and show how to use that to general a successful attack sequence. The differences between two systems are defined as,
\begin{equation}
\begin{split}
     \Delta m_t \triangleq \tilde{m}_t - m_t, \quad \Delta z_t \triangleq \tilde{z}_t - z_t,\quad \Delta r_t \triangleq \tilde{r}_t - r_t.
 \label{eqn:difference}
\end{split}
\end{equation}
\begin{definition}
Given $\delta>0$, the $\chi^2$ detector is successfully attacked if there exists an attack sequence $\epsilon_1,\epsilon_2,\cdots,\epsilon_T$ such that the following holds:
$$||\Delta r_t|| < \delta,\quad \forall t,$$
where $\Delta r_t$ is defined above.
\end{definition}}
\Revised{
\begin{remark}
If $\Delta r_t$ is bounded, then the difference of its quadratic form   $\tilde{g}_t - g_t$ is also bounded. Also, as pointed out by \cite{mo2010control}, by linearity, we can find a $\delta'>0$, such that $|P(\tilde{g}_t>threshold) -P({g}_t>threshold)| \le \delta', \quad \forall t$.
\end{remark}
This definition of successful attack follows the $(\epsilon,\alpha)$--attack definition by Mo et al.~\cite{mo2010control}. When the probability of the alarm $P(\tilde{g}_t>threshold) -P({g}_t>threshold)$ is bounded and  a small enough $\delta$, the alarm rate $\delta'$ will converge to the false alarm rate of the healthy system. Mo et al.~\cite{mo2015performance} presented the relationship between $\delta$ and $\delta'$. 
}

\Revised{
Given the threshold for the $\chi^2$ detector and $\delta$, the question is how to set desired separations $d_1,d_2,\cdots$  such that we can avoid being detected. In the following, we give a sufficient condition for designing $d_1,d_2,\cdots$.
\begin{theorem}
If the separations $ \Delta m_{t+1}$ and $\Delta m_{t}$ satisfy $|| K_{t+1}^{-1}|| \cdot \left\lVert \Delta m_{t+1} - \mathcal{F}\Delta {m}_t \right\rVert \le \delta$, then the proposed algorithm can successfully attack the $\chi^2$ detector at step $t+1$.
\label{the:detector}
\end{theorem}
}
\begin{proof}
\Revised{
Manipulating Equations (\ref{eqn:kalman1}), (\ref{eqn:kalman2}), (\ref{eqn:kalman3}), and (\ref{eqn:difference}), we can prove that,
\begin{equation}
    \Delta m_{t+1} =   K_{t+1} \Delta r_{t+1}  + \mathcal{F}\Delta {m}_t.
    \label{eqn:DeltaZt}
\end{equation}
Taking the norm of Equation~(\ref{eqn:DeltaZt}), we have
\begin{equation}
\begin{split}
      || \Delta r_{t+1}||  &= \left\lVert -K_{t+1}^{-1}\Delta m_{t+1}  + K_{t+1}^{-1} \mathcal{F}\Delta {m}_t \right\rVert\\
      &\le || K_{t+1}^{-1}|| \cdot \left\lVert \Delta m_{t+1} - \mathcal{F}\Delta {m}_t \right\rVert.\\
\end{split}
\end{equation}
Therefore,
\begin{equation}
    ||\Delta r_{t+1}|| \le || K_{t+1}^{-1}|| \cdot \left\lVert \Delta m_{t+1} - \mathcal{F}\Delta {m}_t\right\rVert.
\end{equation}
We apply the condition of successful attack. If we have,
\begin{equation}
\begin{split}
        || K_{t+1}^{-1}|| \cdot \left\lVert \Delta m_{t+1} - \mathcal{F}\Delta {m}_t \right\rVert \le \delta
\end{split}
\label{eqn:successful_attack}
\end{equation}
then, 
$$||\Delta r_k|| < \delta.$$
Note that $\mathcal{F}$ is a known matrix, and the Kalman gain $K_t$ can be computed from the initial covariance matrix $\Sigma_0$. 
Hence, we can design the attack sequence $\epsilon_1, \epsilon_2,\cdots$ for a $\chi^2$ detector given the threshold $\delta$.}
\end{proof}

\Revised{
Theorem~\ref{the:detector} shows that if we want to attack a system with $\chi^2$ detector, the strategy is to make the difference between two consecutive desired separations, $d_{t+1}$ and $||\mathcal{F}||d_t$, as small as possible. In general, when we design the attack sequence, we want to increase the separation to mislead the system. Without loss of generality, we can consider the case that all the elements in  $m_{t} - \mathcal{F}\Delta {m}_{t-1}$ are non-negative. Given a known separation from previous step $t-1$, we have the following condition for $d_t$ when we design the desired separation:
 \begin{equation}
     d_{t} -|| \mathcal{F}\Delta m_{t-1}|| \le \delta, \quad t>1.
     \label{eqn:detector_condition}
 \end{equation}
\begin{remark}
Applying Theorem~\ref{the:detector} and Equation~(\ref{eqn:detector_condition}) and given $\delta$, we can design a sequence of separations $d_1,d_2,\cdots,d_T$ a priori since $d_t = ||\Delta m_t||$. For example, if we know the Kalman filter's initial condition, assuming $d_1=0$, we have,
\begin{equation}
        d_2 = \delta ||K_1||,
    \label{eqn:Computedt1}
\end{equation}
With a known $\Delta m_2$ from the proposed LP algorithm, $d_3$ can be designed with the following equation,
\begin{equation}
\begin{split}
            d_{3} = (|| \mathcal{F}&\Delta m_{2}|| + \delta)\cdot ||K_1||.
    \label{eqn:Computedt2}
\end{split}
\end{equation}
Iteratively, we can get the desired separation for all times and guarantee a successful attack ($||\Delta r_t|| < \delta,\quad t=1,2,\cdots ,T.$)\footnote{Since the inequality is conservative, if the equation converges to $d_{t+1}=d_{t}$, we can add a small term $O(t)$, and let $d_{t} = (|| \mathcal{F}\Delta m_{t-1}|| + \delta)\cdot ||K_t||+O(t)$.}.
\end{remark}
}
 In the following section, we will provide an example that by increasing the separation with given condition. The simulation shows the $\chi^2$ will not alarm when the separation is designed as Theorem~\ref{the:detector}.
\subsection{Simulation with $\chi^2$ detector}
\Revised{We consider the $L_1$  vector norm and the same model from the simulation section. We use the following parameters:
$$ \mathcal{F}=I_{2\times2},  \mathcal{G}=I_{2\times2},  
  R =0.1I_{2\times2}, 
    Q =0.1I_{2\times2}, \delta =0.1.\\$$
Given $\delta =0.1$, we can design the separation $d_t$. The attack result and the $\chi^2$ detector value are shown in Figure~\ref{fig:Successful_attack_chisquare}.}
\begin{figure}[H]
\centering
\includegraphics[width=1\columnwidth]{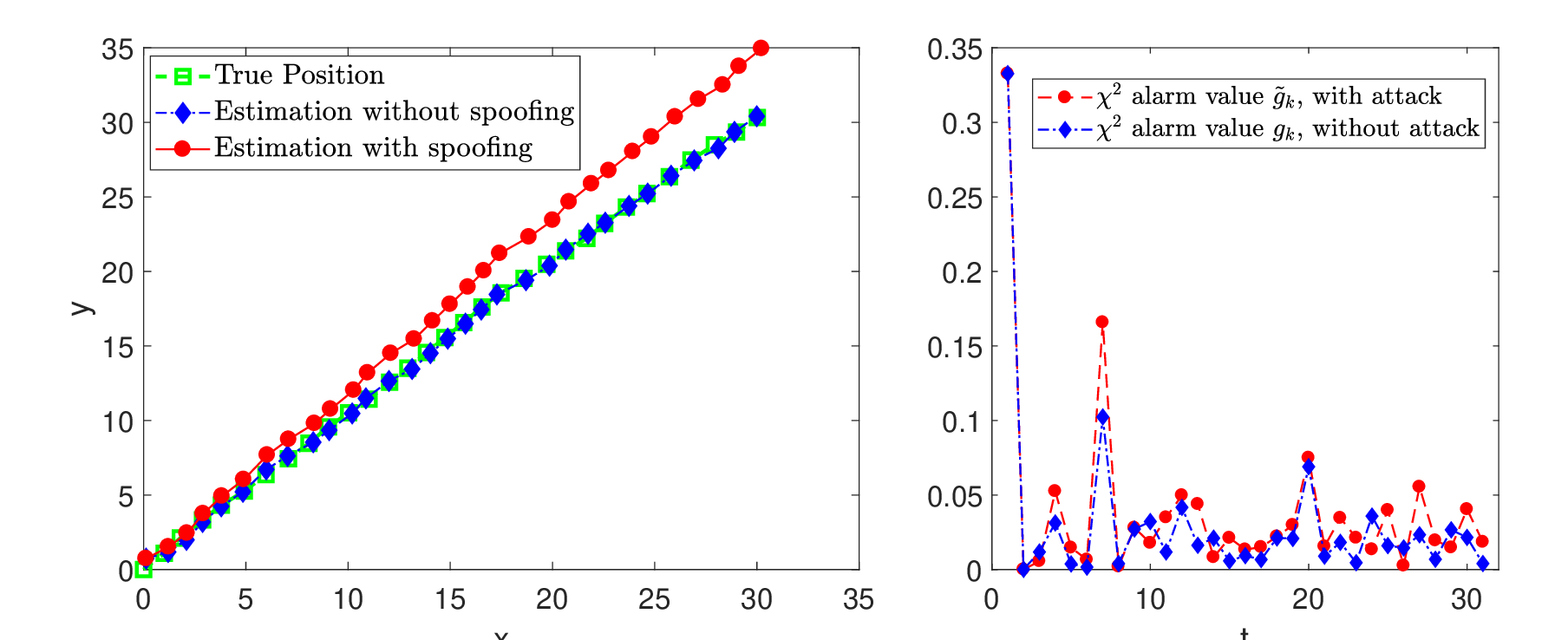}
\caption{Estimated positions and the $\chi^2$ detector's output ($g_k$), when the spoofing signals are injected by setting $d_t$ based on Theorem~\ref{the:detector}.}
\label{fig:Successful_attack_chisquare}         
\end{figure}

\Revised{
\begin{figure}[H]
\centering
\includegraphics[width=0.7\columnwidth]{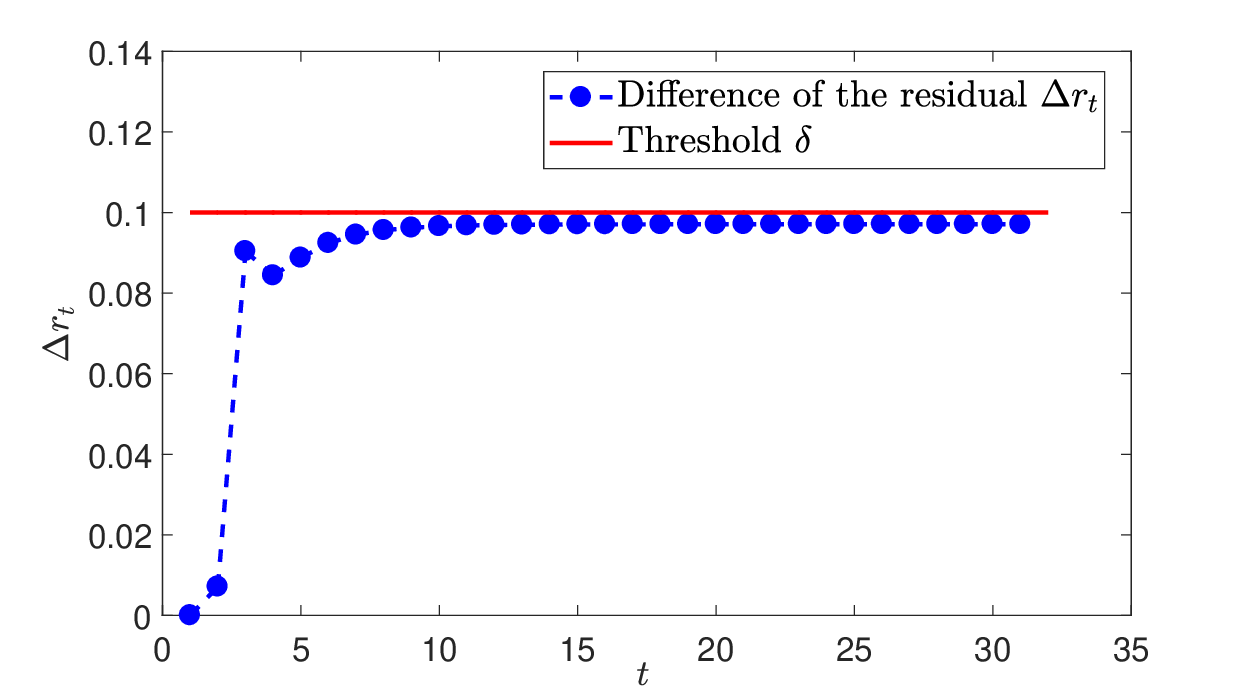}
\caption{ The  differences of  the two residual $||r_t - \tilde{r}_t||$. The threshold $\delta = 0.1$.}
\label{fig:redisual_threshold}         
\end{figure}
The differences ($\Delta r_t$) between the residual are shown in Figure~\ref{fig:redisual_threshold}.
We ran 1000 trials using this strategy as shown in Figure~\ref{fig:Successful_attack_chisquare}. The $\chi^2$ detector detected the attack in 112 trials (The false alarm rate is equal to $11.054\%$, this rate indicates the $\chi^2$ detector can not tell whether the alarm is a false alarm or not). This is close to the actual false alarm rate without any attack. In this scenario, the system will not be able to distinguish between false and true alarms. Thus, the attack strategy is able to successfully mislead the $\chi^2$ detector. }

\section{Conclusion} \label{sec:conc}

We study the problem of injecting spoofing signals to achieve a desired separation in the output of a Kalman filter without and with attack. We study many variants of the problem. Our main approach was to formulate the problems as nonlinear, constrained optimization problems in order to minimize the energy of the spoofing signal. We show that under some technical assumptions, the problems can be solved by linear programming optimally. We present a more computationally expensive approach to solve the problem, without the aforementioned assumptions. \Revised{We also present a sufficient condition for this strategy to mislead the $\chi^2$ failure detector}. 

Our immediate future work is to study the game-theoretic aspects of the problem. In this work, we did not consider any active strategy being employed by the observer to mitigate the attack. In future works, we will consider the case of designing spoofing signals that explicitly take the attack mitigation strategies into account. In all the problems considered in this paper, the desired separations are taken as inputs provided by the user. The simulation results suggest that carefully choosing a specific profile of the desired separation can make it harder to detect by the observer. A possible extension is to automatically generate the optimal profile that not only minimizes the signal energy but also ensures that it is not detected by the observer. Another future work is to extend the strategy to more general non-linear state estimation approaches, such as the extended Kalman filter, unscented Kalman filter, and particle filters. 
\section*{APPENDIX}

\subsection{Proof of Theorem \ref{mk}}
\label{sec:proof_difference_mt}
Before we prove Theorem~\ref{mk}, we can review the Kalman Filter update equations from equation (\ref{eqn:difference}), (\ref{eqn:kalman1}), (\ref{eqn:kalman2}), (\ref{eqn:kalman3}).

According to the Kalman gain update Equation~(\ref{eqn:kalman3}), the evolution covariance matrix at step $t$, $\Sigma_t$, only depends on the state model parameters and the initial condition of the covariance matrix $\Sigma_0$. The Kalman gain at step $t$, $K_t$ depends on the covariance matrix $\Sigma_t$. Both $\Sigma_t$ and $K_t$ do not depend on the control input series $\{u_t\}_{t=1,\cdots,k}$, measurement $\{z_t\}_{t=1,\cdots,k}$. Thus, the covariance matrix and the Kalman gain can be predicted from the KF covariance update steps.
\begin{equation} \label{Riccati}
\begin{split}
&{\Sigma}_{t+1|t}=\mathcal{F}{\Sigma}_{t|t}\mathcal{F}'+R_t,\\
&{\Sigma}_{t+1|t+1}= (I- K_t \mathcal{H}) {\Sigma}_{t+1|t}.
 \end{split}
\end{equation}

From Equation \eqref{Riccati}, the Kalman gain can be predicted from the initial condition $\Sigma_0$. 

We now prove our main result.

\begin{proof} 
From the update of  KF, we have
\begin{equation} \label{kalman_covariance}
\begin{split}
 m_t& = m_{t|t-1} + K_t(z_t-\mathcal{H}m_{t|t-1})\\
&=(I-K_t\mathcal{H})m_{t|t-1} +K_t z_t\\
&=(I- K_t\mathcal{H}) (\mathcal{F}m_{t-1}
+\mathcal{G}u_{t-1}) +K_t z_t.
\\
\end{split}
\end{equation}
and
$$\tilde{m}_t = (I- \tilde{K}_t\mathcal{H}) (\mathcal{F}\tilde{m}_{t-1} +\mathcal{G}u_{t-1}) +\tilde{K}_t (z_t + \epsilon_t).$$
Recursively,
\begin{equation}
\begin{split}
\label{Proposition3}
m_t & - \tilde{m}_t \\
=& (I- K_t\mathcal{H}) (\mathcal{F}m_{t-1} +\mathcal{G}u_{t-1}) +K_t z_t \\
&- [(I- \tilde{K}_t\mathcal{H}) (\mathcal{F}\tilde{m}_{t-1} +\mathcal{G}u_{t-1}) +\tilde{K}_t (z_t+\epsilon_t)]\\
=& (\mathcal{F} - K_t \mathcal{H} \mathcal{F})m_{t-1} - (\mathcal{F}- \tilde{K}_t \mathcal{H} \mathcal{F})\tilde{m}_{t-1}\\
&-(K_t - \tilde{K}_t)\mathcal{H} \mathcal{G} u_{t-1} + [K_t z_t - \tilde{K}_t(z_t+ \epsilon_t)],
\end{split}
\end{equation}
subtract a term  $\tilde{K}_t \mathcal{H}\mathcal{F}m_{t-1}$ then add the same term, 
\begin{equation}
\begin{split}
\label{Proposition2}
m_t & - \tilde{m}_t \\
=& (\mathcal{F} - \tilde{K}_t\mathcal{H}\mathcal{F})m_{t-1} - (\mathcal{F}- \tilde{K}_t\mathcal{H}\mathcal{F})\tilde{m}_{t-1}\\
&- (K_t - \tilde{K}_t) \mathcal{H} \mathcal{G} u_{t-1} +(K_t - \tilde{K}_t)z_t - \tilde{K}_t \epsilon_t\\
&-(K_t - \tilde{K}_t) \mathcal{H}\mathcal{F}m_{t-1}\\
=&(\mathcal{F} - \tilde{K}_t \mathcal{H}\mathcal{F})(m_{t-1} - \tilde{m}_{t-1})\\
&+(K_t - \tilde{K}_t)[z_t - \mathcal{H} (\mathcal{F}m_{t-1} + \mathcal{G}u_{t-1})] -\tilde{K}_t \epsilon_t.
\end{split}
\end{equation}

Define, $A_t = \mathcal{F} - \tilde{K}_t \mathcal{H}\mathcal{F}$, $B_t = (K_t-\tilde{K}_t)[z_t-\mathcal{H}(\mathcal{F}m_{t-1}+\mathcal{G}u_{t-1})]$ and $C_t = - \tilde{K}_t \epsilon_t.$
Then,
\begin{equation*}
\begin{split}
\label{Proposition2}
m_t  -& \tilde{m}_t \\
=& A_t (m_{t-1} - \tilde{m}_{t-1}) + B_t + C_t\\
=& A_t[A_{t-1}(m_{t-2} - \tilde{m}_{t-2}) + B_{t-1} + C_{t-1}] + B_t + C_t\\
\qquad \dots \\
=&\prod_{i=0}^{t-1}A_{t-i}\cdot (m_0  - \tilde{m}_0)+ (B_t + C_t)\\
&+ A_t (B_{t-1}+C_{t-1}) +  \cdots
+ A_t \cdots A_3 A_2 (B_1 + C_1)  \\
= &\prod_{i=0}^{t-1}A_{t-i}\cdot (m_0  - \tilde{m}_0) +B_t + C_t\\
&+\sum_{i=0}^{t-2} \left( \prod_{j=0}^{i}A_{t-j} (B_{t-1-i} + C_{t-1-i}) \right) .
\end{split}
\end{equation*}
\end{proof}

\bibliographystyle{IEEEtran}
\bibliography{IEEEabrv,main}

\end{document}